\documentclass[11pt]{article}
\usepackage{geometry}
\usepackage{amssymb}
\usepackage{mathrsfs}
\usepackage{amsmath}
\usepackage{times}
\usepackage{amsmath}
\usepackage{graphicx}

\newtheorem{theorem}{Theorem}[section]

\newtheorem{proposition}[theorem]{Proposition}

\def\squarebox#1{\hbox to #1{\hfill\vbox to #1{\vfill}}}
\newcommand{\qed}{\hspace*{\fill}
\vbox{\hrule\hbox{\vrule\squarebox{.667em}\vrule}\hrule}\smallskip}

\newenvironment{proof}{\noindent{\bf Proof:~~}}{\(\qed\)}

\begin{document}
%%%%%%%%%%%%%%%%
\title{Efficiently Enumerating Scaled Copies of Point Set Patterns}

\author{Aya Bernstine\footnote{
School of Computer Science and Engineering , The Hebrew University, Jerusalem, Israel. aya.bernstine@mail.huji.ac.il}   
\hspace{1mm} and Yehonatan Mizrahi\footnote{
School of Computer Science and Engineering , The Hebrew University, Jerusalem, Israel. yehonatan.mizrahi@mail.huji.ac.il
}
}
% \coauthor{}
\date{}
\maketitle

\maketitle

\begin{abstract}

Problems on repeated geometric patterns in finite point sets in Euclidean space are extensively studied in the literature of combinatorial and computational geometry. Such problems trace their inspiration to Erd\H{o}s' original work on that topic. In this paper, we investigate the particular case of finding scaled copies of any pattern within a set of $n$ points, that is, the algorithmic task of efficiently enumerating all such copies. We initially focus on one particularly simple pattern of axis-parallel squares, and present an algorithm with an $O(n\sqrt{n})$ running time and $O(n)$ space for this task, involving various bucket-based and sweep-line techniques. Our algorithm is worst-case optimal, as it matches the known lower bound of $\Omega(n\sqrt{n})$ on the maximum number of axis-parallel squares determined by $n$ points in the plane, thereby solving an open question for more than three decades of realizing that bound for this pattern. We extend our result to an algorithm that enumerates all copies, up to scaling, of any full-dimensional fixed set of points in $d$-dimensional Euclidean space, that works in time $O(n^{1+1/d})$ and space $O(n)$, also matching the corresponding lower bound due to Elekes and Erd\H{o}s.

\end{abstract}

\section{Introduction}

The problems of geometric point pattern matching and the identification of repeated geometric patterns are fundamental computational problems with a myriad of applications, ranging from computer vision \cite{inproceedings2, MOUNT199917}, image and video compression \cite{article2}, model-based object recognition \cite{inproceedings3}, structural biology \cite{article5} and even computational chemistry \cite{inproceedings4}. Such problems were originally motivated by questions from extremal combinatorics regarding the maximal number of occurrences of a given pattern determined by a set of points, a field historically inspired by Erd\H{o}s' well-known Unit Distance Problem (1946) regarding the maximal number of unit distance pairs induced by such sets \cite{Erds1946OnSO}. %Due to their importance and applicability, these computational problems have been approached along the years using numerous ways, including tree pruning [REF6], geometric hashing [REF7], statistics [REF8], neural networks [REF9] and many more.
Our paper approaches the computational problems of identifying patterns using tools and techniques encountered in the framework of computational geometry, ensuring exact, provably correct and efficient solutions. For a detailed survey that discusses the applicability of these methods and their relevance to such problems, see \cite{article17}.

In this paper, we analyze a particular case of interest -- the problem of identifying shifted and scaled copies of  \textit{any} point set pattern in Euclidean space, termed homothetic copies of this set. For brevity, we refer to those as copies of the pattern. Note that each copy is certainly determined by the images of any two points in the set.
We begin with focusing on one of the simplest forms of these problems, due to its historical significance, where the goal is to identify repeated patterns of axis-parallel squares in the plane. Simply put, we strive for an efficient algorithm that lists all subsets of four points that constitute the vertices of a square with axis-parallel edges, an algorithm that could be easily generalized to treat the identification of any pattern in any dimension. Note that coping with the challenge of enumerating all copies of a given pattern trivially encompasses a solution for the task of finding one single copy of that pattern in a background set of points, namely the task of pattern matching. This manifests the importance of dealing with the more general of those two challenges, thereby solving both at one fell swoop.

This problem was extensively studied in the literature of combinatorial and computational geometry. As articulated in 1990 by van Kreveld and de Berg \cite{10.1007/3-540-52292-1_25}, the maximum possible number of axis-parallel squares determined by $n$ points in the plane is $\Theta(n\sqrt{n})$, and given any set of $n$ points, all such squares can be enumerated in time $O(n\sqrt{n} \log n)$\footnote{The analysis given throughout this paper of time and space complexities is based on the conventional and very standard word-RAM model of computation \cite{fredman1990blasting}.} and space $O(n)$. They consolidate their result by showing an extension of their algorithm to one that enumerates all full-dimensional axis-parallel hypercubes in $d$-dimensional Euclidean space in time $O(n^{1+1/d}\log n)$, with a logarithmic-factor gap separating this computational result from the corresponding combinatorial geometric bound of a maximum of $\Theta(n^{1+1/d})$ possible hypercubes, raising the challenge of overcoming that gap as an open question. The latter combinatorial result was further extended by Elekes and Erd\H{o}s \cite{elekes1994similar}. They established a bound of $\Theta(n^{1+1/d})$ on the maximum number of copies of \textit{any} pattern in $\mathbb R^d$, assuming the pattern is full-dimensional and has rational coordinates. The computational aspect of this generalization occurs in \cite{Bra2002CombinatorialGP}, providing an algorithm that works in time $O(n^{1+1/d}\log n)$ for the task of enumerating all copies of any fixed full-dimensional set consisting of a constant number of points in $d$-dimensional Euclidean space, exhibiting the same logarithmic-factor gap between the two results. We remark that the complexity of finding copies in these algorithms becomes smaller with increasing dimension, a rare occurrence.

\subsection{Our Results}

Our main result of this paper is an efficient deterministic algorithm that enumerates all scaled copies of any given pattern consisting of a constant number of points in $d$-dimensional Euclidean space.
As earlier stressed, we first handle the case of listing all axis-parallel squares in the plane. Although being a special case that follows immediately from our main theorem, there are several motivations for its individual treatment. First, it arguably forms the simplest symmetric full-dimensional pattern in the plane. Consequently, it earned its right of having the first attempt, as we are aware of, among all problems of listing scaled copies of some pattern, of admitting an improved performance compared to the trivial algorithm, in terms of running time analysis. Adaptations to other patterns appeared, e.g., in \cite{Bra2002CombinatorialGP}, but \cite{10.1007/3-540-52292-1_25} were the first to raise the question of whether it is computationally feasible to realize the combinatorial bound of $\Theta(n\sqrt{n})$ possible axis-parallel squares, thereby improving their result having the running time of $O(n\sqrt{n} \log n)$. Our algorithm which is, in some sense, an amplified version of their original algorithm, is therefore a \textit{worst-case optimal} solution for this problem, and fully addresses this question which was open for more than three decades.

A second motivation for separately handling this case, is that it results in a relatively natural algorithm having less obfuscating technicalities as compared to the general case of treating arbitrary patterns. In other words, it crystallizes the additional techniques we incorporate within the original solution of \cite{10.1007/3-540-52292-1_25}, yielding the required improvement. Briefly sketched, the gist of these techniques includes the following ideas. First, we use a reduction from arbitrary input points to points having ''compressed'' coordinates, that is, we relabel their values, allowing the use of linear non-comparative sorting methods on the result. Second, we deploy a sweep-line sub-procedure that marks points forming a square, instead of searching those in a set, avoiding the corresponding logarithmic factor without involving probabilistic tools as hashing. Third, in order to carry out the mentioned sweep-line process correctly, we also relabel the sum and the difference of the input coordinates, in addition to the relabeling of the coordinates themselves. We show why the last step is crucial for the algorithm to succeed in the full elaboration of our solution in Section 2.

\vspace{3mm}

\noindent \textbf{Theorem:}
    \textit{Given a planar set of points $P$ of size $n$, all axis-parallel squares defined by points from $P$ can be enumerated in time $O(n\sqrt n)$ and $O(n)$ space.}
\vspace{3mm}

% \vspace{3mm}

% \noindent \textbf{Theorem 1:}
%     \textit{Given a planar set of points $P$ of size $n$, all axis-parallel squares defined by points from $P$ can be enumerated in time $O(n\sqrt n)$ and $O(n)$ space.}
% \vspace{3mm}

As previously stressed, this algorithm matches, in the worst-case, the known bound of $\Theta(n\sqrt{n})$ on the maximum possible number of axis-parallel squares.
We further extend this result to a companion proposition, that discusses the task of reporting all full-dimensional hypercubes in $d$-dimensional Euclidean space. A solution for the latter task that works in time $O(n^{1+1/d})$ and space $O(n)$, which also matches the corresponding known lower bound and addresses a similar open question, is provided.

Establishing our main result, which is concerned with finding scaled copies of \textit{any} arbitrary and full-dimensional pattern in $d$-dimensions, relies on the majority of the ideas behind the proof of the previous theorem. Assuming the pattern consists of a constant number of points, we provide an algorithm with a running time of $O(n^{1+1/d})$. Similarly to the proof idea of the previous theorem mentioned earlier, to accomplish this task we relabel some particular and carefully chosen affine transformations of the input coordinates along with the coordinates themselves, a relabeling that creates a convenient representation of the points, for the purpose of scanning them using a sweep-line technique.

\vspace{3mm}

\noindent \textbf{Theorem:}
    \textit{Given a fixed set of points $Q$ of full dimension in the $d$-dimensional Euclidean space, and a set of points $P$ of size $n$, all scaled copies of $Q$ defined by points from $P$ can be enumerated in time $O(n^{1+1/d})$ and $O(n)$ space.}
\vspace{3mm}

The running time in this theorem matches the corresponding combinatorial bound of the same magnitude, and improves the best known running time for this task of $O(n^{1+1/d}\log n)$ \cite{Bra2002CombinatorialGP}.
We provide the proof of the last theorem in Section 3.

Note that although the improvement suggested is by a logarithmic factor, the upshot is a \textit{worst-case optimal} algorithm in terms of running time analysis, even for the most general case of arbitrary patterns. This can be compared with a work by de Rezende and Lee \cite{article6}, who studied the problem of enumerating all \textit{rotated} copies of a given pattern, improving the running time of the trivial algorithm for this task by a logarithmic factor as well. Other than that, aside from the worst-case optimality of our results, the techniques deployed form a rather general scheme, and may therefore be potentially useful to treat other variants of the problem studied.

\subsection{More Related Work}

%In \cite{article6}, de Rezende and Lee studied a problem of a similar flavor to ours, namely that of enumerating all copies of a given pattern, only in their case -- copies obtained by rotation, instead of scaling. They devised an algorithm with a running time of $O(n^d)$ in $d$-dimensions, assuming $d>3$, improving the running time of the trivial algorithm by a logarithmic factor.
An excellent survey that covers a variant of our problem, that deals with rotated copies instead of scaled ones, along with combinations of rotations and scaling, affine transformations of patterns, Ramsey-type graph-theoretic related problems and more can be found in \cite{Bra2005Chapter2P}.

A result related to this variant, namely that of rotated copies, by Cho and Mount \cite{article115} considers the problem of approximated pattern matching allowing only rotations. They showed an efficient algorithm that admits a $3.13$-approximation to the optimum. Namely, it returns a transformation whose Hausdorff distance is at most a factor $3.13$ larger than the optimal Hausdorff distance from the pattern.

In \cite{10.1007/3-540-52292-1_25}, the problem of counting the number of all axis-parallel rectangles (boxes in $d$-dimensions) is also considered, providing an algorithm with a running time of $O(n^{2-1/d})$. Note that as opposed to enumerating axis-parallel squares, it is an increasing function of $d$.

\subsection{Paper Organization}
The rest of the paper is organized as follows. In Section 2 we provide an algorithm that efficiently enumerates all axis-parallel squares. An elaboration on the techniques used for the speeding-up previous results is also given. Section 3 generalizes the result from Section 2 to full-dimensional hypercubes, and ultimately to an efficient enumeration scheme of all copies of any arbitrary point set pattern. Concluding remarks and directions for further research are given in Section 4.

%\newpage

\section{Axis-Parallel Squares}
The purpose of the current section is to present an efficient algorithm that reports all axis-parallel squares defined by a planar set of points. That is, given $n$ points on the plane, the algorithm addresses the task of listing all subsets of four points that form the vertices of an axis-parallel square. Obviously, this task can be tackled by a trivial algorithm that sorts the input points, say, according to their $x$ coordinate, and then scans every pair of points residing in the same column and determines whether that pair can be complemented to a square from the right, by searching in a set, for instance. However, the performance of this algorithm is rather poor in terms of complexity analysis, as it results in a running time of $O(n^2\log n)$, e.g., when all points share the same $x$ coordinate.

A more efficient algorithm, devised by van Kreveld and de Berg \cite{10.1007/3-540-52292-1_25},
manages to improve upon this by first treating only columns with at most $\sqrt n$ points, in an identical fashion to that described in the trivial algorithm, then deleting those columns, and finally applying a similar procedure on the remainder, only considering rows instead of columns. In this section, we refer to columns with at most $\sqrt n$ points as ''short'' columns.
The full description of their algorithm is as follows:

\vspace{3mm}
\textbf{Squares-Listing}$(p_1,...,p_n)$:
\begin{enumerate}
    \item Build a balanced search tree $T$ and an array $A$ on the input points, both sorted by the $x$ coordinate.
    \item For every pair of points $p$ and $q$ in $A$ residing in a short column, search in $T$ whether they can be complemented to a square from the right or from the left. Report each square found, unless the other two vertices defining it are on a short column to the left of $p$ and $q$, to avoid reporting the same square more than once.
    \item Delete all short columns from $T$ and $A$, and convert each remaining point $(x,y)$ to $(y,x)$.
    \item Apply step 2 on the remaining converted points.
\end{enumerate}

% \begin{theorem} (van Kereveld, de Berg)
%     \label{thm:KB}
%     Given a planar set of points $P$ of size $n$, all axis-parallel squares defined by points from $P$ can be enumerated in time $O(n\sqrt n \log n)$ and $O(n)$ space.
% \end{theorem}
\vspace{3mm}
The correctness of the algorithm is proven in their work, along with its complexity analysis. They show that it operates in a running time of $O(n\sqrt n \log n)$ and $O(n)$ space.
% The full proof of Theoerm \ref{thm:KB}, which relies on the above algorithm, can be found in [REF]. 
As we integrate most of its ideas in our amplified version, we present the analysis of the relevant steps in detail in the proof of our theorem in the current section, later on.

We strive for an algorithm that does not include the above logarithmic factor, i.e., one with a running time of $O(n\sqrt{n})$ and space $O(n)$. As shown in \cite{10.1007/3-540-52292-1_25}, there are configurations of $n$ points that define $\Theta(n\sqrt{n})$ squares, or even more generally:
\begin{theorem} (van Kreveld, de Berg)
    For a set $P$ of $n$ points in d-dimensional space, the maximum possible number of $2^d$ points that form the vertices of an axis-parallel hypercube is $\Theta(n^{1+1/d})$.
\end{theorem}

This theorem obviously induces a lower bound on the running time of the optimal algorithm for the enumeration task. Thus, our result in this section bridges the logarithmic-factor gap between this lower bound, and the previously best known upper bound for this objective, attained by Squares-Listing$(p_1,...,p_n)$.

\subsection{Main Ideas Towards an Improvement}

The logarithmic factor in the running time of Squares-Listing$(p_1,...,p_n)$ is present due to the step in which we search points in a set.
There are several steps we deploy in order to overcome this. Assume first, for simplicity, that all input points have coordinates in $\{1,...,n\}$. Then, instead of searching in a set the two \textit{query} pairs that complement the pair $(x,y),(x,y+\delta)$ to a square, i.e., the pair $(x+\delta,y),(x+\delta,y+\delta)$ and the pair $(x-\delta,y),(x-\delta,y+\delta)$, we put all query points (points being part of some query pair) along with the original input points in an array, apply radix sort on that array, treating each point as a two-digit number, and mark all positive queries. That is, we mark each query point adjacent to an existing input point sharing the same coordinates, or to an already marked identical query point. The resulting marked pairs define the appropriate existing squares. Using this non-comparative sorting algorithm, along with its linear scanning that follows, we obtain a process that does not involve the logarithmic factor of marking a query point.

Generally, however, we cannot assume that all coordinates are taken from $\{1,...,n\}$. We address this issue by ''shrinking'' the coordinates of all input points. In other words, we label each $x$ coordinate in the input with a value in $\{1,...,n\}$, and likewise for the $y$ coordinate, so effectively we reduce the problem to the aforementioned restricted version. The main caveat of this idea, though, is that arithmetical considerations regarding labels are invalid. For example, if we map the pair of points $(10,20)$ and $(10,40)$ to the labeled points $(1,2)$ and $(1,4)$ respectively, the pre-translated pair may be complemented to a square using the pair $(30,20)$ and $(30,40)$, but there is no guarantee that $(30,20)$, if exists in the input, was mapped to $(3,2)$.

To overcome the above issue, we avoid using arithmetic considerations when defining a query pair of points $q_1,q_2$ that complement the post-labeled pair (i.e., the points after being renamed) $p_1=(x,y),p_2=(x,y+\delta)$ to a square (from the right, assuming that $\delta>0$, for example). That is, instead of using the invalid label $x+\delta$ as a coordinate, we define the pair $q_1,q_2$ using \textit{identical} labels as those of $p_1,p_2$ in the following manner: Note that $q_2$, aside from sharing the same $y$ coordinate as $p_2$ (assuming $q_2$ lies above $q_1$, for instance), it also shares the same Manhattan distance from the origin as $p_1$\footnote{This is true assuming all points have positive coordinates, but we can safely assume this, as otherwise we can shift the input so it maintains that form, without affecting our result.
}, namely they have the same sum of coordinates. The full elaboration of this consideration appears in the course of the correctness proof of our algorithm. Therefore, if we label also each \textit{sum of coordinates} in the input points, similarly to the labeling of the coordinates themselves, we can search for that exact label and avoid the arithmetic considerations. In short -- the query point $q_2$ is defined having the same $y$ label as $p_2$ and the same $x+y$ label as $p_1$. Similarly, we label also the differences of the coordinates, so $q_1$ has the same $y$ label as $p_1$ and the same $y-x$ label as $p_2$, as we show in the proof of our following theorem. Note that the label of $x+y$ need not be the sum of the labels of $x$ and $y$ in this approach. Similar considerations apply for the pair that complements to a square from the left.

Another observation related to the idea described in the last paragraph, which will be useful in further developments of the current section's algorithm, is that rotating any vector $(x,y)$ in the plane by $45^{o}$ and stretching it by a factor of $\sqrt{2}$ yields the vector $(x+y,y-x)$, as illustrated in Figure 1. This technique can be shown to establish, for example, the famous geometric equivalence between the $l_1$ and $l_{\infty}$ norms (see, e.g., \cite{article187}), namely that the $l_1$ distance between any two pre-rotated points equals the $l_{\infty}$ distance between the post-rotated points, after scaling. For our purposes, we note that instead of viewing the process as labeling the $x+y$ and $y-x$ values of the coordinates, as suggested above, one can think of this process as the labeling of the post-rotated points' coordinates -- in the case of axis-parallel squares, a rotation by the aforementioned factor.

\begin{figure}
    \centering
    \includegraphics[width=14cm]{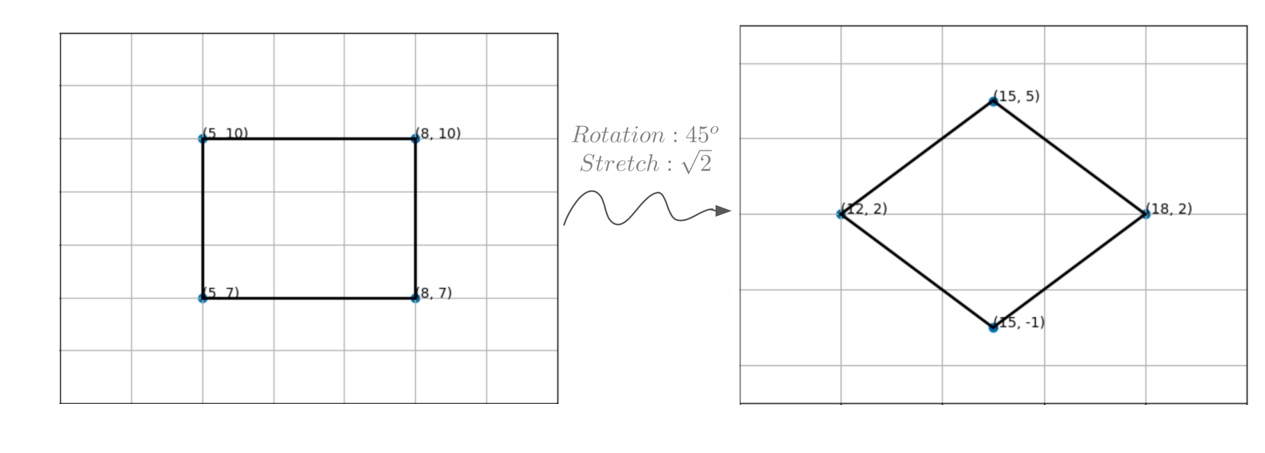}
    \caption{Illustrating the rotation by $45^{o}$ and the stretch by a factor of $\sqrt{2}$ applied on four points in the plane. Each point $(x,y)$ was converted to the point $(x+y,y-x)$ as a result.}
    \label{fig:galaxy}
\end{figure}
%These three bucketing-based ideas, namely the use of counting sort and marking..

\subsection{The Efficient Solution}

These ideas from the previous subsection lead to our main theorem of this section:

\begin{theorem}
    Given a planar set of points $P$ of size $n$, all axis-parallel squares defined by points from $P$ can be enumerated in time $O(n\sqrt n)$ and $O(n)$ space.
\end{theorem}

\begin{proof}
    Assume w.l.o.g. that all points have positive coordinates, as otherwise we can shift those without affecting the result.
    The theorem's statement is accomplished using the following algorithm:
    
    \vspace{3mm}    
    \textbf{Amplified-Squares-Listing}$(p_1,...,p_n)$:
    \begin{enumerate}
        \item Change the representation of each point $p=(x,y)$ to the representation $(x,y,x+y,y-x)$. \\ Map each $x$ coordinate in the input to a value in $\{1,...,n\}$. Perform a similar procedure for the $y$ coordinates, the $x+y$ coordinates and the $y-x$ coordinates.
        % \item 
        \item Build an array $A$ on the input points, sorted by the $x$ coordinate.
        %, intended for scanning pairs of input points. %Build two additional arrays, $B_1$ and $B_2$, both initialized to contain only the input points, such that $B_1$ is sorted . 
        \item For each pair of post-labeled points $p_1=(x,y_1,x+y_1,y_1-x)$ and $p_2=(x,y_2,x+y_2,y_2-x)$ with $y_2>y_1$, out of the first $n$ pairs of points in $A$ that reside in a short column -- define the query pair  $q_1=(*,y_1,x+y_2,*),q_2=(*,y_2,*,y_1-x)$ that complements $p_1,p_2$ to a square from the right, and the query pair $q_1'=(*,y_1,*,y_2-x),q_2'=(*,y_2,x+y_1,*)$ that complements to a square from the left. The wildcards replace the two unknown coordinates, which are uniquely defined anyway, given the other two.
        % cool bro
        
        \item Place each query point defined by its $y$ and $x+y$ coordinates, i.e., points of the form of $q_1$ and $q_2'$ from step 2, in an array $B_1$ along with all input points, and apply radix sort on $B_1$ based on those two coordinates. Perform a similar procedure for points of the form of $q_2$ and $q_1'$ in another array $B_2$.
        
        \item Scan $B_1$ and mark each query point adjacent to an input point sharing the same coordinates, or to an already marked identical query point. Act similarly on $B_2$. Report each square found during this scanning, unless the other two vertices defining it are on a short column and complement to a square from the left, to avoid reporting the same square more than once.
        
        \item Perform steps $3$-$5$ iteratively on each subsequent $n$ pairs of points in $A$ that reside in a short column.
    
        \item Delete all points that are on short columns from $A$. Convert each remaining point $(x,y)$ to $(y,x)$. Apply steps $1$-$6$ on the remaining converted points.

    \end{enumerate}
    
    \textbf{Correctness}: Most of the main ideas behind the algorithm's correctness were described earlier, in the previous subsection. We now elaborate on the rest of the details. All squares having at least one edge on a short column are reported in step 5 of the algorithm, before applying step 6. The rest have both edges on long columns, so they are reported in step 7.
    
    Given a pair $p_1=(x_1,y_1)$ and $p_2=(x_1,y_2)$ with $y_2>y_1$, the pair $q_1=(x_2,y_1)$ and $q_2=(x_2,y_2)$ that complement to a square from the right (i.e., $x_2>x_1$) maintains that $x_2+y_1=x_1+y_2$ since they lie on a diagonal with the same Manhattan distance from the origin. Putting if differently,
    $$x_2=x_1+(y_2-y_1) \Longrightarrow x_2+y_1=x_1+(y_2-y_1)+y_1=x_1+y_2$$
    As for $q_2$, the latter equality also shows that $x_2-y_2=x_1-y_1$ by subtracting $y_1+y_2$ from both sides. These are exactly the query points defined by the algorithm, up to the labeling that maintains those properties. The analysis for the pair that complements to a square from the left is symmetric.
    
    \vspace{3mm}
    
    \textbf{Complexity}: As for the running time, the first two steps of the algorithm cost $O(n\log n)$ using some standard sorting algorithm, e.g., merging sort. Each time step 3 is performed, at most $2n$ query points are defined in $O(n)$ time. Each time steps 4-5 are performed, two arrays, each of size at most $n$, are sorted and then scanned in a linear time. Defining the obtained squares in step 5, based on the marked queries, can also be carried out in $O(n)$ time. The total number of query points defined after finishing step 6 is
    $$O\left(\sum_i s_i^2\right)\leq O\left(\sum_i s_i\sqrt{n}\right)=O\left(\sqrt{n}\cdot \sum_i s_i\right) \leq O\left(n\sqrt{n}\right)$$
    where $s_i$ denotes the length of the $i$'th short column, i.e., the number of points in it. As mentioned, each batch of $O(n)$ queries is handled in $O(n)$ time, so the total running time analysis for steps 1-6 of this algorithm is $O(n\sqrt{n})$. The analysis for the converted points in step 7 is symmetric. It only remains to notice that the number of pairs, this time, is 
    $$O\left(\sum_i d_i^2\right)\leq O\left(\sum_i d_i\sqrt{n}\right)=O\left(\sqrt{n}\cdot \sum_i d_i\right) \leq O\left(n\sqrt{n}\right)$$
    where $d_i$ denotes the length of the $i$'th row out of the remaining rows, after deleting the short columns. We used the fact that there are at most $\frac{n}{\sqrt{n}}$ long columns, as otherwise there are more than $n$ input points. Therefore, the length of each remaining row, after deletion, is at most $\frac{n}{\sqrt{n}}=\sqrt{n}$, and all points are treated. We remark that parts of this analysis occur in \cite{10.1007/3-540-52292-1_25}, but we provide those in our proof for completeness.
    
    As for the space complexity, note that each of the data structures defined in the above algorithm is of size $O(n)$, and that each step involving those structures does not cost more than $O(n)$ space.
\end{proof}

\section{The General Case}

In this section, we treat the generalized version of the problem from the previous section. We describe an algorithm that enumerates all scaled copies of any arbitrary full-dimensional pattern in $d$-dimensions, assuming the pattern is of a fixed size. Since such copies are determined by the images of any two points in the pattern, a trivial listing algorithm defined similarly to the trivial algorithm from the previous section can be easily obtained. We can safely assume w.l.o.g. that there are two pattern points that share a coordinate, as otherwise we can rotate the input. The trivial algorithm merely chooses these two pattern points and calculates in $O(1)$ the transformations that produce each of the other pattern points. Then it scans each pair of the background points that share this coordinate, and verifies whether they could be complemented to that pattern using the computed transformations. The running time of this algorithm is therefore $O(n^2\log n)$.

In \cite{10.1007/3-540-52292-1_25}, the authors managed to generalize their axis-parallel squares listing algorithm to an algorithm that enumerates all axis-parallel full-dimensional hypercubes in $d$-dimensions, outputting the list of all hypercubes in time $O(n^{1+1/d}\log n)$. Their algorithm relies on similar ideas that guided the treatment of the planar case, namely the partitioning to short axis-parallel lines (columns, in the planar case), handling those first. In \cite{Bra2002CombinatorialGP}, the latter result was further generalized to any full-dimensional pattern in $d$-dimensions, providing an algorithm that lists all copies of it, and that admits the same running time complexity.

% This section is built in a modular manner. We first devise an efficient algorithm that treats the hypercubes pattern. The main ideas of it rely on the techniques and the observations from the previous section, with some additions and adjustments so it complies with the properties of the $d$-dimensional space. We then generalize those adjustments so they fit the framework of general patterns.

\subsection{Warm Up - Axis-Parallel Hypercubes}
We first treat the case of enumerating all full-dimensional axis-parallel hypercubes. We provide an algorithm with a running time complexity of $O(n^{1+1/d})$ and with a linear space complexity, addressing the open question of matching the lower bound from Theorem 1. Our algorithm builds on the techniques and the observations from the previous section, with some additions and adjustments so it complies with the properties of the $d$-dimensional space.

One observation that is true for the $d$-dimensional case, is that any two points with all but one equal coordinate, determine $2^{d-1}$ full-dimensional possible hypercubes. Denote such a pair of points by $t$ and $r$. Each of these hypercubes is uniquely associated with a vector $e\in \{-1,1\}^{d-1}$, in which the $j$'th coordinate determines the direction of progress from $t$ and $r$ along the $j$'th axis, where $j$ is any coordinate except the one in which they differ.
In a similar fashion to the planar case, we would like to relabel the input coordinates, their sums and their differences, place them in an array, radix sort it and mark the correct vertices that complement to a hypercube in an efficient manner, while scanning this array. Moreover, we scan only pairs of points lying on short axis-parallel lines, similarly to the planar case, only that this time, by ''short'' we mean having not more than $n^{1/d}$ point on it.

Keeping these ideas in mind, we are ready to present the generalized version of Theorem 2:

\begin{proposition}
    Given a set of points $P$ of size $n$ in $d$-dimensional Euclidean space, all axis-parallel full-dimensional hypercubes defined by points from $P$ can be enumerated in time $O(n^{1+1/d})$ and $O(n)$ space.
\end{proposition}

\begin{proof}
    Assume w.l.o.g. that all points have positive coordinates, as otherwise we can shift those without affecting the result.
    The following algorithm establishes the proposition's statement:
    
    \vspace{3mm}

\textbf{Amplified-Hypercubes-Listing}$(p_1,...,p_n)$:
\begin{enumerate}
    \item For each input point $p=(x_1,x_2,...,x_d)$, add the following additional list of coordinates: $$\left((x_i-x_j),(x_i+x_j) ~~|~~\forall 1\leq i<j\leq d\right)$$
    Map each of those augmented coordinates, including the original ones, to a label in $\{1,...,n\}$.
    \item Build an array $A$ on the input points, sorted by each of their coordinates based on the coordinates' order, except for the last \textit{original} coordinate (i.e., $x_d$).
    \item For each pair of points $t,r$ that lie in the same short axis-parallel line, having the same coordinates except for the last, out of the first $n$ pairs with this property, add $2^{d}-2$ query points which define together a hypercube. Do this for all $2^{d-1}$ possible hypercubes in the following manner.
    First, any axis-parallel hypercube having $t$ and $r$ as its vertices, is defined using one additional vertex 
    $$r'=(r_1+e_1\cdot \delta,...,r_{d-1}+{e_{d-1}} \cdot \delta,t_d)$$
    where $r_i$ is the $i$'th coordinate in $r$ (and similarly for $t$), $\delta=r_d-t_d$ and $e\in \{-1,1\}^{d-1}$. The rest of the vertices in each such hypercube are defined similarly, except for replacing all subsets of the coordinates in the vector $e$ by zeros, and using $t$ instead of $r$.
    
    Now, define the coordinates that are to be searched -- \textit{not} in the aforementioned arithmetic manner, but using the labels from step 1 instead.
    %In fact, as in the previous section, the algorithm does not search for the coordinates defined in this arithmetic manner, but using labels. 
    That is, translate $r_i+e_i\cdot \delta$ to the label of $r_i+r_d$ if $e_i=1$, and to that of $r_i-r_d$ otherwise. Fill in the unknown coordinates using wildcards, as those are uniquely defined anyway, given the others.
    
    % we think of each $t_i+e_i\cdot \delta$ coordinate as the label we translate it to, namely $t_i-r_i$ if $e_i=1$ and $t_i+r_i$ if $e_i=-1$. %For each coordinate we translate  its value to a label in $\{1,...,n\}$.
    
    \item
    Place all query points defined by the same coordinates in an array along with all input points. Apply radix sort on each of those arrays, according to the known coordinates in it.
    \item Scan each array from step 4, and mark each query point adjacent to an input point sharing the same coordinates, or to an already marked identical query point. Report all hypercubes that were found (by checking that all vertices are present for each hypercube), except for hypercubes that have two vertices on a short axis-parallel line of the same type, only with a smaller index.
    \item Perform steps 3-5 on each subsequent $n$ pairs of points in $A$ on a short axis-parallel line of that type.
    \item Delete all points that are on a short axis-parallel line of a currently analyzed type from $A$, and convert each remaining point $(x_1,x_2,...,x_d) $ to $(x_d,x_1,...,x_{d-1})$. Apply steps 1-6 on the remaining converted points. This step is carried out $d-1$ times.
    \end{enumerate}
    
    \textbf{Correctness}:  Almost all details regarding the analysis of the correctness of this algorithm already appeared in that of Amplified-Squares-Listing$(p_1,...,p_n)$. As for the phase of searching by labels, note that if $r_j'$ is some unknown coordinate, for which we only have an undesired arithmetic definition based on the coordinate $r_j$ and on $\delta$, then
    $$r_j' =r_j+\delta=r_j+(r_d-t_d) \Longrightarrow r_j'+t_d=r_j+r_d$$
    which exactly corresponds to the labels that the algorithm searches (the treatment of positive or negative values of $\delta$ is symmetric), and similarly for subtraction.

    \vspace{3mm}
    
    \textbf{Complexity}: The running time analysis is similar to that of Amplified-Squares-Listing$(p_1,...,p_n)$ with the following differences. The relabeling in step 1 and the sorting of $A$ in step 2; the definitions of $2^{d-1}$ hypercubes, each consisting of additional $2^d-2$ points in step 3; applying radix sort on $2^{O(d)}$ arrays in step 4; the linear scanning of those arrays in step 5 and applying step 7 for $d-1$ times -- all of those involve multiplying the complexity of the planar case by at most a constant factor of $2^{O(d)}$. The only major difference concerns the total number of query points defined each time step 7 in the algorithm is invoked. In the treatment of the first $d-1$ axis-parallel lines, only short lines are considered, so the cost for each line is
    $$O\left(\sum_i s_i^2\right)\leq O\left(\sum_i s_i\cdot n^{1/d}\right)=O\left(n^{1/d}\cdot \sum_i s_i\right) \leq O\left(n^{1+1/d} \right)$$
    where $s_i$ denotes the length of the $i$'th short axis-parallel line of that form.
    Regarding the treatment of the last axis-parallel line, we note that all remaining such lines are short. Assume that there exists a long remaining line. Then all points on it are on long axis-parallel lines, with respect to some axis, which induces more points that are on long lines with respect to another axis, yielding that there are more than $\left(\left(n^{1/d}\right)\right)^d$ input points in total, which is obviously a contradiction. Thus, all input points are treated with the mentioned running time, so the total running time analysis is indeed $O(n^{1+1/d})$. The analysis of the space complexity is similar, with a constant multiplicative factor of $2^{O(d)}$ compared to the analysis of the planar case.

\end{proof}

\subsection{General Scaled Copies of a Pattern}
 
At this point, we have all the necessary building blocks to present our final result of this paper, that deals with enumerating all full-dimensional copies of a given fixed-size pattern. We show an algorithm for this task that works in time $O(n^{1+1/d})$ with a linear space complexity. This answers the open question of realizing the lower bound of Elekes and Erd\H{o}s \cite{elekes1994similar} regarding the maximum possible number of occurrences of such patterns, raised in \cite{Bra2002CombinatorialGP}. In the following result, we rely on an observation from Section 2, by which we can generate identical labels by means of rotating the input points, which is equivalent to using the sum (or difference) of the coordinates in the simple case of hypercubes.

The full statement of our theorem is as follows:
\begin{theorem}
    Given a fixed set of points $Q$ of full dimension in the $d$-dimensional Euclidean space, and a set of points $P$ of size $n$, all scaled copies of $Q$ defined by points from $P$ can be enumerated in time $O(n^{1+1/d})$ and $O(n)$ space.
\end{theorem}

\begin{proof}
    Assume w.l.o.g. that all points have positive coordinates, as otherwise we can shift those without affecting the result. Assume also that the pattern has more than two points, as the other case may be easily handled.
    The following algorithm is used to prove the theorem:
    
    \vspace{3mm}

\textbf{Amplified-Scaled-Copies-Listing}$(p_1,...,p_n)$:
\begin{enumerate}
    \item Rotate the input points and the pattern points simultaneously, such that two of the pattern points, $p$ and $q$, lie in the same column afterwards (in $d$-dimensions -- in some axis-parallel line). Calculate the transformations that produce the other pattern points out of $p$ and $q$.
    
    \item For a point $r\neq p,q$ in $Q$, calculate the rotation that places $r$ and $p$ on the same column. Thus, $p$ and $r$ share the same value of some linear transformation of their coordinates. Label the original input points' coordinates, along with that linear transformation value for those, using labels in $\{1,...,n\}$ (in $d$-dimensions -- perform this labeling procedure in all dimensions, as the pattern is full-dimensional).
    
    \item Build an array $A$ on the input points, sorted by their $x$ coordinate (in $d$-dimensions -- a sorting based on all coordinates except for one).
    
    \item For each pair $r$ and $t$ residing in a short column (in $d$-dimensions -- short axis parallel line that corresponds to the sorting from step 3), calculate all values of the relevant pattern points according to the transformations from step 1. Express those using the corresponding labels obtained from step 2, until defining $n$ such query set. If there are other pairs among the pattern points determining a line parallel to that through $r$ and $t$, let $r$ and $t$ have also the role of that pair when defining the other points that complement to the pattern, in a similar manner to the treatment of the hypercubes case.
    
    \item Place all query points defined by the same (augmented) coordinates in an array with all input points. Apply radix sort on each such array.
    
    \item Scan each array from step 5, and mark each query point adjacent to an input point sharing the same coordinates, or to an already marked identical query point. Report all copies that were found, except for copies that are duplicate (may be generated by two pairs in the pattern having the same slope).
    \item Perform steps 4-6 on each subsequent $n$ pairs of points in $A$ on a short column (in $d$-dimensions -- axis-parallel line of that type).
    \item Delete all points that are on a short column (in $d$-dimensions -- axis-parallel line of a currently analyzed type) from $A$. Apply steps 1-6 on the remaining points, choosing this time a different pair which is not parallel to the previous pairs chosen. This step is carried out $d-1$ times.    
\end{enumerate}
    
    \textbf{Analysis}: The correctness and the complexity analysis of this algorithm are essentially the same as those of Theorem 2 and Proposition 3, with some minor exceptions. The labeling of the points, and searching by those labels, are both legitimate, as justified earlier. The first two steps cost $O(1)$, as the pattern is of fixed size. The analysis from the proof of Proposition 3, by which at each iteration (dimension), the number of treated points is $O(n^{1+1/d})$ is similar here, as we did not use there the fact that the lines are axis-parallel, but rather only the fact that the corresponding vectors form an independent set. Thus, the argument by which there would be more than $\left(n^{1/d}\right)^d$ points, if the last iteration includes a long axis-parallel line, still holds.
\end{proof}

\section{Conclusion and Further Work}

In this paper, we analyzed the problem of efficiently enumerating all axis-parallel squares within a set of $n$ points in the plane. We provide an algorithm with a running time of $O(n\sqrt{n})$ for this task, thereby solving an open question for more than three decades \cite{10.1007/3-540-52292-1_25} of matching the corresponding lower bound on the running time of the optimal algorithm for this task, as the maximum possible number of such squares is $\Theta(n\sqrt{n})$. We consolidate our result by treating the $d$-dimensional case of listing all full-dimensional hypercubes in $\mathbb R^d$ in time $O(n^{1+1/d})$. Next, we show that our framework can even solve the general case of reporting all scaled copies of a given pattern of fixed size in time $O(n^{1+1/d})$, answering an open question from \cite{Bra2002CombinatorialGP} by matching a similar lower bound due to Elekes and Erd\H{o}s \cite{elekes1994similar}. All of our algorithms admit a linear space complexity.

As for the techniques used, we relied on some existing algorithms for those tasks, amplifying and speeding up their running time using bucket-based methods (non-comparative sorting, labeling), sweep-line scanning and more. As far as we are aware of, the combinations of these techniques in the way we deployed was not noted in the literature so far for dealing with tasks having a similar flavor.

One immediate open question is whether these can be applied or adjusted for other pattern matching problems, such as copies obtained from rotating a patters, combining rotations and scaling and more, solving some of the open problems in this field. Another possibility is applying them even for different computational geometric problems. Another questions is whether the task of finding a single copy of a general pattern is an easier computational task than enumerating all copies of it. For the case where the pattern is not of a constant size, but rather input dependent, the best known algorithm \cite{Bra2002CombinatorialGP} includes a multiplicative factor of the magnitude of the pattern size, as compared to the combinatorial bound on the maximum possible number of its copies, that does not include that factor. Thus, also shaving this factor from the computational results becomes another important open question for further research.

\subsubsection*{Acknowledgements}

We would to thank the students in the Algorithmic Problem Solving class (2019/2020) at the Hebrew University of Jerusalem for bringing to our attention that this was a known open problem in the field of computational geometry.

\bibliographystyle{plain}
\bibliography{patterns}

\appendix

\end{document}